\documentclass[letterpaper, 10 pt, conference]{ieeeconf}  

\IEEEoverridecommandlockouts                              

\usepackage{graphicx} 
\usepackage{caption,xcolor}

\usepackage{amsmath}
\usepackage{amssymb}
\usepackage{amsfonts}
\let\labelindent\relax
\usepackage{enumitem}

\newcommand{\vu}{\boldsymbol{u}}

\newcommand{\vx}{\boldsymbol{x}}

\newcommand{\mxa}{\boldsymbol{A}}
\newcommand{\mxb}{\boldsymbol{B}}

\newcommand{\mxu}{\boldsymbol{U}}

\newcommand{\mxw}{\boldsymbol{W}}

\newcommand{\R}{\mathbb{R}}

\newcommand{\fig}[1]{Fig.~\ref{fig:#1}}

\newcommand{\mat}[1]{{\begin{pmatrix} #1 \end{pmatrix}}}
\newcommand{\eqna}[1]{\begin{equation}\begin{aligned}#1 \end{aligned}\end{equation}}

	\newtheorem{lemm}{Lemma}
			
\title{\LARGE \bf
Measuring LTI System Resilience against Adversarial Disturbances\\
based on Efficient Generalized Eigenvalue Computations
}
\author{Johannes Börner$^{1}$ and Florian Steinke$^{2}$
\thanks{*This work was sponsored by the German Federal Ministry of Education and
	Research in project AlgoRes, grant no. 01jS18066A. It has been performed
	in the context of the LOEWE center emergenCITY}
\thanks{$^{1}$Johannes Börner is with Faculty of Electrical Engineering, 
        Technical University Darmstadt, 64283 Damrstadt, Germany
        {\tt\small johannes.boerner@eins.tu-darmstadt.de}}%
\thanks{$^{2}$Florian Steinke is with Faculty of Electrical Engineering, 
	Technical University Darmstadt, 64283 Damrstadt, Germany
        {\tt\small  florian.steinke@eins.tu-darmstadt.de}}%
}
\usepackage[mode=buildnew]{standalone}
\usepackage{tikz}
\usetikzlibrary{calc,patterns,angles,quotes,decorations.pathmorphing}

\usepackage[europeanresistors,americaninductors]{circuitikz}
\newcommand\height{4}
\newcommand\hheight{2}
\newcommand\diameter{0.4}
\newcommand\yh{1}

\tikzset{
	pics/carc/.style args={#1:#2:#3}{
		code={
			\draw[pic actions] (#1:#3) arc(#1:#2:#3);
		}
	}
}
\usepackage{diagbox,tabularx,tabu}
\usepackage{booktabs,ragged2e}
\usepackage{caption, subcaption}
\usepackage[absolute]{textpos}
\newcommand{\copyrightstatement}{
		\noindent
		\footnotesize
		\copyright  2021 IEEE Personal use of this material is permitted. Permission from IEEE must be obtained for all other uses, in any current or future media, including reprinting/republishing this material for advertising or promotional purposes, creating new collective works, for resale or redistribution to servers or lists, or reuse of any copyrighted component of this work in other works.
}    
\newcommand\copyrightnotice{%
	\begin{tikzpicture}[remember picture,overlay]
		\node[anchor=south,yshift=10pt] at (current page.south) {\fbox{\parbox{\dimexpr\textwidth-\fboxsep-\fboxrule\relax}{\copyrightstatement}}};
	\end{tikzpicture}%
}
\usetikzlibrary{calc,patterns,angles,quotes,decorations.pathmorphing}

\tikzset{
	pics/carc/.style args={#1:#2:#3}{
		code={
			\draw[pic actions] (#1:#3) arc(#1:#2:#3);
		}
	}
}
\begin{document}


\maketitle

\begin{abstract}
Resilient systems are able to recover quickly and easily from disturbed system states 
that might result from hazardous events or malicious attacks.
In this paper a novel resilience metric for linear time invariant systems is proposed: 
the minimum control energy required to disturb the system
is set into relation to the minimum control energy needed to recover.
This definition extends known disturbance rejection metrics considering random effects to account for adversarial disturbances.
The worst-case disturbance and the related resilience index can be computed efficiently via solving a generalized eigenvalue problem
that depends on the controllability Gramians of the control and disturbance inputs.
The novel metric allows improving system resilience by optimizing the restorative control structure or by hardening the system against specific attack options.
The new approach is demonstrated for a coupled mechanical system.
\end{abstract}

\section{Introduction}
\copyrightnotice

The concept of a resilience triangle \cite{Bruneau2003} has been influential in the study of the resilience of (complex networked) systems.
The triangle is the area between the actual degraded system performance after a disturbance and the desired quality of service integrated over time.
The area depends both on the robustness of the system, i.e., that the system performances does not degrade too much following a certain disturbance caused by failure or attack, 
as well as the time for restoring the system afterwards.
For power systems, the concept has been widely used and adapted \cite{Roege2014,2016Hosseini,Bhusal2020}, 
measuring the quality of service via the load not served \cite{Panteli2015}, the number of failed network components \cite{Panteli2017a}, or grid frequency deviations \cite{2017Bie,Boerner2019a}.

Such resilience measures are dependent on the examined contingencies.
These may, however, not be known exactly in advance
and it is thus worthwhile to analyze more general, structural system properties.
Structural controllability theory \cite{lin1974structural,Liu2016} works based on the 
the connectivity graph between the dynamic states and the system inputs 
as implied by the sparsity structure of the system matrices.
It allows making statements about the structural controllability of the system and to efficiently minimize the number of nodes needed to control the system \cite{Liu2011,Pequito2016}.
Using the control path length as an indicator for the required control energy, the choice of control nodes can further be optimized \cite{li2018enabling}.
When taking the actual values of the system matrices and the implied system dynamics into account,
the degree of controllability can be measured based on the spectrum of the controllability Gramian of the system \cite{Muller1972}.
Suitable control nodes can be selected by optimizing different norms of this Gramian via efficient sub-modular techniques \cite{Summers2016}.
The approach also allows to characterize graph structures in synchronization problem that can be distorted with minimum energy \cite{Dhal2016}.
The optimal selection of control nodes for static linear systems with constraints is considered in \cite{mora2021minimal}.
While these approaches measure the general degree of controllability, they do not consider the interplay between a specific type of disturbance and the counteracting control system.

This interaction is be considered when computing the degree of disturbance rejection \cite{Kang2009}, 
i.e., the average energy to counteract colored random noise via the trace of a matrix involving the controllability Gramian of the system and the covariance matrix of the noise source.
The work shares some similarity to our approach but does not consider worst-case or adversarial behavior of the disturbance source.
This, however, is of major interest in many networked systems which are vulnerable to cyber attacks \cite{case2016analysis, Dhal2016, chen2020study}.

In this paper we propose to model a specific type of distortions, the \emph{attacker}, and the type of available control system, the \emph{defender}, via two input matrices for a linear time invariant dynamical system.
We can then analyze the ratio of the control energies needed for various attack and defense trajectories.
If both sides behave optimally, assuming full model knowledge on both sides, we can compute the critical energy ratio via a generalized eigenvalue problem involving the controllability Gramians of the attacker and the defender.
We propose to use the resulting value as a novel resilience index.
For a mechanical test system we show that the index yields plausible resilience ratings, both when the assumptions used to derive the index are exactly fulfilled as well as when these assumptions are partially violated.

The proposed novel metric takes the type of actions available to the attacker and the defender into account, 
including the possible effects resulting from the system dynamics and the interplay between the two.
It does not focus on individual contingencies.
Moreover, in terms of the resilience triangle the metric includes both elements of robustness, the required energy for distorting the system, as well as the time and ease for restoring the system, the defense energy.
The examined adversarial setting is especially valuable when considering malicious distortions such as cyber attacks. 
The index value is also efficient to compute.

The remainder of the paper is structured as follows: 
In Section~\ref{sec:prel} some necessary results from minimum energy control are reviewed and our notation is introduced.
The new metric is derived step by step in Section~\ref{sec:met}.
It is demonstrated for a system of three coupled pendula in Section~\ref{sec:ver}.
We conclude in Section~\ref{sec:con}.

\section{Preliminaries}\label{sec:prel}

This section defines our notation and reviews required results from minimum energy control theory \cite{Antsaklis1997}.
Consider the linear time invariant (LTI) system 
\eqna{\label{eq:prob}
	\dot{\vx}=\mxa\vx+\mxb\vu
}
with state $\vx(t) \in \R^n$, control input $\vu(t) \in \R^m$, and system matrices $\mxa,\mxb$ of appropriate sizes.
Let $(\mxa,\mxb)$ be controllable and let $\vu(t;\vx_0,\vx_1,t_0,t_1)$ denote a control input that moves the system state from $\vx_0$ at time $t_0$ to $\vx_1$ at time $t_1$.
The control energy of any control $\vu(t)$ during time span $[t_0,t_1]$ is defined as
\eqna{
E[\vu(t),t_0,t_1] = \int\limits_{t_0}^{t_1} \vu^T(\tau)\vu(\tau)d\tau.
}
Of all controls $\vu(t;\vx_0,\vx_1,t_0,t_1)$ let $\vu^*(t;\vx_0,\vx_1,t_0,t_1)$ denote the one that minimizes this control energy.
Using the symmetric positive definite controllability Gramian
\eqna{
	\mxw(t_1,t_0)=\int\limits_{t_0}^{t_1}e^{\mxa (t_1-\tau)} \mxb\mxb^T e^{\mxa^T (t_1-\tau)}d\tau,
}
this minimum energy control $\vu^*(t;\vx_0,\vx_1,t_0,t_1)$ can be computed in closed form as 
\eqna{\label{eq:opt_u}
\mxb^Te^{\mxa^T (t_1-t)}\mxw^{-1}(t_1,t_0)\Delta\vx_{(\vx_0,\vx_1,t_0,t_1)},
}
where $\Delta\vx_{(\vx_0,\vx_1,t_0,t_1)} = \vx_1-e^{\mxa (t_1-t_0)}\vx_0$.
The achieved minimum energy $E[\vu^*(t;\vx_0,\vx_1,t_0,t_1),t_0,t_1]$ then is
\eqna{\label{eq:opt_E}
\Delta\vx^T_{(\vx_0,\vx_1,t_0,t_1)}\mxw^{-1}(t_0,t_1)\Delta\vx_{(\vx_0,\vx_1,t_0,t_1)}.
}

For $t_1 = t_0$ the controllability Gramian is $\mxw(t_0,t_0) = 0$ and its time derivative is 
\eqna{\label{eq:Wc_deriv}
\frac{d}{dt_1}\mxw(t_1,t_0) = \mxa\mxw(t_1,t_0)+\mxw(t_1,t_0)\mxa^T + \mxb\mxb^T.
}
This allows computing the Gramian for all time intervals $[t_0,t_1]$ via a linear initial value problem.
For stable system matrices $\mxa$, the Gramian will converge and $\dot{\mxw}(\infty) = 0$ implies the well-known Lyapunov equation $\mxa\mxw(\infty)+\mxw(\infty)\mxa^T = -\mxb\mxb^T$.

If $(\mxa,\mxb)$ is not controllable, there exists a subspace $V\subseteq\R^n$ of the state space that cannot be reached by the controller.
It holds that $\vx^T\mxw(t_1,t_0)\vx=0$ for $\vx\in V$, i.e. the Gramian is only positive semi-definite and not invertible in this case.
The symmetric positive semi-definite Gramian can be decomposed as $\mxw(t_1,t_0)=\mxu diag(\lambda_1,..,\lambda_k,0,\ldots,0)\mxu^T$, $\lambda_1,\ldots,\lambda_k > 0$ and $\mxu$ orthogonal.
With slight abuse of typical notation we will in the following use $\mxw(t_1,t_0)^{-1}$ to denote $\mxw(t_1,t_0)^{-1}=\mxu diag(1/\lambda_1,..,1/\lambda_k,M,\ldots,M)\mxu^T$ for the non-invertible case, where $M$ is a very large number.
Note that this definition is different to the typical pseudo-inverse, but implies that the minimum energy to reach a state in $V$ is very large according to \eqref{eq:opt_E}. 
For reachable states orthogonal to $V$ the energy expression \eqref{eq:opt_E} still yields valid results.

\section{Proposed Resilience Metric}\label{sec:met}

To measure a system's resilience we consider a LTI system with two inputs
\eqna{\label{eq:sys2}
	\dot{\vx}=\mxa\vx+\mxb_a\vu_a + \mxb_d\vu_d.
}
The attacker (malicious agent, technical failure, or other unknown external influence) has control over the attack input $u_a(t)$ and the defender uses inputs $\vu_d(t)$.
Whereas the defender should have full control of the system, i.e., $(\mxa,\mxb_d)$ is controllable, the attacker might only be able to influence a few points in the system, i.e., $(\mxa,\mxb_a)$ will often not be controllable and the extended formalism introduced above will apply.

Now, imagine that the attacker first moves the system from its original state $\vx_0$ to state $\vx_1$ during time span $[t_0,t_1]$ and the defender then steers the system back to state $\vx_0$ until time $t_2$.
We propose to measure the resilience of the system as the ratio of the control energies associated to the two input trajectories,
\eqna{\label{eq:ratio}
  \frac{E[\vu_a(t;\vx_0,\vx_1,t_0,t_1),t_0,t_1]}{E[\vu_d(t;\vx_1,\vx_0,t_1,t_2),t_1,t_2]}.
}
If the attacker can create a system state deviation with little energy compared to the defender's efforts the system is not resilient.
If the defender only needs little energy to undo the attacker's actions the system is resilient.

\bigskip
Assuming an adversarial attacker, it will try to minimize the system's resilience and thus use a minimum energy attack trajectory. 
On the contrary, the defender will try to maximize the resilience and thus also use a minimum energy defensive strategy.
The critical value of the resilience measure then is
\eqna{\label{eq:ratio2}
  \frac{E[\vu^*_a(t;\vx_0,\vx_1,t_0,t_1),t_0,t_1]}{E[\vu^*_d(t;\vx_1,\vx_0,t_1,t_2),t_1,t_2]}\\
	= \frac{\Delta\vx^T_{(\vx_0,\vx_1,t_0,t_1)}\mxw_a^{-1}(t_0,t_1)\Delta\vx_{(\vx_0,\vx_1,t_0,t_1)}}
	{\Delta\vx^T_{(\vx_1,\vx_0,t_1,t_0)}\mxw_d^{-1}(t_1,t_2)\Delta\vx_{(\vx_1,\vx_0,t_1,t_2)}}
}
where $\mxw_a$ / $\mxw_d$ denote the controllability Gramian related to the attack / defensive input matrices $\mxb_a$ / $\mxb_d$.

\bigskip
Let's now assume that $\mxa$ is stable, either since the underlying system is naturally stable or since it is stabilized by a controller and we only consider closed loop system dynamics.
Moreover, let $\vx_0=0$ be the equilibrium point. The resilience measure then simplifies to 
\eqna{\label{eq:ratio21}
\frac{\vx_1\mxw_a^{-1}(t_0,t_1)\vx_1}
	{\vx_1 \underbrace{e^{A^T(t_2-t_1)} \mxw_d^{-1}(t_1,t_2) e^{A(t_2-t_1)}}_{=\tilde{\mxw}_d^{-1}(t_1,t_2) } \vx_1}.
}
An adversarial attacker will choose the attack vector $\vx_1$ that minimizes this ratio.
This yields our final resilience index $\rho(t_0,t_1,t_2)$,
\eqna{\label{eq:ratio3}
\rho(t_0,t_1,t_2) = \min_{\vx_1} \frac{\vx_1\mxw_a^{-1}(t_0,t_1)\vx_1}
	{\vx_1 \tilde{\mxw}_d^{-1}(t_1,t_2) \vx_1},
}
which can be computed as the minimum extended eigenvalue of $\mxw_a^{-1}(t_0,t_1)$ and $\tilde{\mxw}_d^{-1}(t_1,t_2)$.

\bigskip
This resilience index can also be computed without inverting the Gramian matrices.
This follows from a lemma presented and proved in the appendix.
It is 
\eqna{\label{eq:ratio4}
\rho(t_0,t_1,t_2) 
&= \min_{\vx_1} \frac{\vx_1 \tilde{\mxw}_d(t_1,t_2) \vx_1}	{\vx_1\mxw_a(t_0,t_1)\vx_1} \\
&= \left(\max_{\vx_1} \frac{\vx_1 \mxw_a(t_0,t_1) \vx_1}{\vx_1 \tilde{\mxw}_d(t_1,t_2)\vx_1}\right)^{-1},
}
i.e., we can compute the resilience index via the maximum extended eigenvalue of $\mxw_a(t_0,t_1)$ and $\tilde{\mxw}_d(t_1,t_2) = e^{A(t_1-t_2)} \mxw_d(t_1,t_2) e^{A^T(t_1-t_2)}$.

Avoiding the inversion is beneficial for two reasons.
First, it reduces the computational costs for computing the inverses.
Second, we avoid actually working with the extended inverse of the Gramian defined above if $(\mxa,\mxb_a)$ is not controllable. 
Instead, the attack Gramian can be used directly, whether invertible or not. 

\bigskip
The proposed resilience metric is not dependent on a specific attack trajectory, but only on the type of a possible distortion, the type of the defensive control system and the system dynamics, 
as characterized through the system matrices $\mxb_a$, $\mxb_d$, and $\mxa$, respectively.
The metric does assume, however, that both parties act optimally.
If this assumption is violated, no strict statements can be derived and the energy ratio may take different values.
For behavior which is outside, but close to the assumptions, we show in our experimental section that the resilience index can still give reasonable ratings about the observed energy ratios.

The resilience index also depends on the time spans $t_1-t_0$ and $t_2-t_1$.
If the time span $[t_1,t_2]$ is significantly longer than the characteristic time scale of the stable system matrix $\mxa$, the system will approach its equilibrium point $\vx_0=0$ without the need for the defender to perform any action at all, independently of the distorted state $\vx_1$.
The resilience index will correspondingly tend towards infinity since $\tilde{\mxw}_d(t_1,t_2) \to \infty$.
If the time spans $t_1-t_0$ and $t_2-t_1$ approach zero, only the input matrices $\mxb_a$ and $\mxb_d$ are relevant for the value of the resilience index.
If $\mxb_d$ has full row rank, a finite limit value is achieved.
Otherwise, the index converges to zero.
The resilience index thus is most interesting for time spans on the order of the characteristic time scale of the system.
In this case we obtain finite, interpretable values, see e.g. the next section, that take into account both the connectivity of the distortions to the system as well as the system dynamics.

The resilience metric allows to compare different defensive settings $\mxb_d$ and choose the most resilient one. 
Different $\mxb_d$ can, for example, imply different control and communication setups in a distributed system.
One can also optimize for the most resilient defensive structure $\mxb_d$, i.e.,
\eqna{\label{eq:ratio20}
\max_{\mxb_d} \rho(t_0,t_1,t_2;\mxb_d).
}
This should, however, be done subject to suitable normalization conditions for $\mxb_d$, such that not the scaling of $\mxb_d$, but only its structure influences the result.

Alternatively, one can identify which attack setting $\mxb_a$ cannot be well-defended against a given fixed defensive control setup.
A counter measure would be to harden the access to the states that $\mxb_a$ addresses, i.e. changing the $\mxb_a$ that is available to the attacker.

\section{Experiments}\label{sec:ver}

We demonstrate our new resilience index for a mechanical system  consisting of three coupled pendula.
We study two settings. 
The first one is exactly in line with the assumptions of Section~\ref{sec:met}.
The second one tests the qualitative validity of the resilience ratings if the assumptions are partially violated.
I.e. we consider the important cases where the defensive controller does not perform a dedicated minimum energy defense just after a fully observed attack, but when it just executes a typical stabilizing controller throughout and makes no specific efforts for attack detection.

\subsection{Simulation Setup}

\begin{figure}
	\center
		\begin{tikzpicture}

\tikzset{
		spring/.style={thick,decorate,
		decoration={zigzag,pre length=1cm,post length=1cm,segment length=6}},
		springdescr/.style={xshift=.5cm}}
	\coordinate(zero)at(0,0);
		\coordinate(first)at(2,0);
				\coordinate(second)at(6,0);
						\coordinate(third)at(10,0);
		
		\draw (0,0) --(12,0)node(ground){};
		\fill[pattern=north east lines] (0,0) rectangle ++(12,0.5); 
	
      \draw[ultra thick] let\p1=(first)in 
      	(first) --  (\x1,-\height) coordinate (head1);
	  \fill[] (head1) circle (\diameter)node[yshift=-1em, below]{$m_1$};

	  \draw[ultra thick] let\p1=(second)in
	   (second) --(\x1,-\height) coordinate (head2);
	  	\fill[] (head2) circle (\diameter)node[yshift=-1em, below]{$m_2$};
	  
	  \draw[ultra thick] 
	  let\p1=(third)in
	  (third) --(\x1,-\height) coordinate (head3);
		  \fill (head3) circle (\diameter)node[yshift=-1em, below]{$m_3$};
		  
	\pic [draw,  ->, "$\theta_3$", angle eccentricity=1.5] {angle = head3--third--ground};
		\pic [draw=red, text=red,right,"$u_{c3}$", <->,  angle eccentricity=1.5,angle radius=1cm] {angle = zero--third--ground};
	
	\pic [draw,  ->, "$\theta_2$", angle eccentricity=1.5] {angle = head2--second--ground};
		\pic [draw=red, text=red,right,"$u_{c2}$", <->,  angle eccentricity=1.5,angle radius=1cm] {angle = zero--second--ground};
	
	\pic [draw,  ->, "$\theta_1$", angle eccentricity=1.5] {angle = head1--first--ground};
		\pic [draw=red, text=red,right,"$u_{c1}$", <->,  angle eccentricity=1.5,angle radius=1cm] {angle = zero--first--ground};
	\draw [spring,shorten <=1.1em, shorten >=1.1em] (head1.west) -- node[below]{$k_{12}$}(head2.west)  ;
	\draw [spring,shorten <=1.1em, shorten >=1.1em] (head2.east) -- node[below]{$k_{23}$}(head3.west);
	\draw[->, thick,red]
		let \p1=(head1) in 
		(\x1,\y1-2.5em)--(\x1+2em,\y1-2.5em)node[midway, below]{$F_{d1}$};
	
	\draw[->, thick,red]
		let \p1=(head2) in 
		(\x1,\y1-2.5em)--(\x1+2em,\y1-2.5em)node[midway, below]{$F_{d2}$};
	
	\draw[->, thick,red]
		let \p1=(head3) in 
		(\x1,\y1-2.5em)--(\x1+2em,\y1-2.5em)node[midway,below]{$F_{d3}$};
	
	\draw[<->,thick,shorten <=1em,shorten >=1em] let\p1=(first), \p2=(second) in 
	(\x1,-\hheight) -- node[midway,below]{$\theta_1-\theta_2$} (\x2,-\hheight) ;	        
		\draw[<->,thick,shorten <=1em,shorten >=1em] let\p1=(second), \p2=(third) in 
	(\x1,-\hheight) -- node[midway,below]{$\theta_2-\theta_3$} (\x2,-\hheight) ;	        
	
	\end{tikzpicture}
	\caption{Test system with three  coupled pendula. Attack forces are applied to the masses of the pendula whereas the defensive control system determines the torques at the pendula's bases.}\label{fig:pendula}
\end{figure}
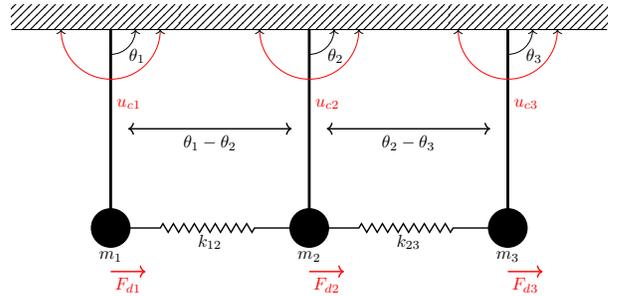

The test system is depicted in \fig{pendula}.
In the linear regime around the steady state, the angular displacements $\theta_i$ of pendula $i=1,2,3$ are governed by
\eqna{
m l \ddot{\theta}_{i}=
	-m g \theta_i -\sum_{j\in N_i} k l (\theta_i-\theta_j)- d_i \dot{\theta}_i  + \frac{1}{l}\tau_{i} + F_i.
}
Here, $m = 1 kg$ is the mass of the pendula, $l = 1m$ their length, $k = 10 N/m$ the spring constant and $g = 10 m/s^2$ the gravitational constant.
$N_i$ denotes the set of neighbors of pendulum $i$ and $d_i$ is a damping factor.
It is chosen as $0.1\times 1/s$ for the left and middle pendulum and $0.3\times 1/s$ for the right pendulum, to break the symmetry between the left and right pendulum.
$\tau_i$ is angular momentum that the defensive controller commands and $F_i$ the attacker's distorting forces.

Choosing as state vector $\vx=\mat{\theta_1\cdots\theta_3,\dot{\theta}_1\cdots\dot{\theta}_3}$, $\vu_c=\mat{\tau_1\cdots\tau_3}$ as the control input, and $\vu_d=\mat{F_1\cdots F_3}$ as the disturbance input 
allows to bring the system into standard form \eqref{eq:sys2}.
The smallest eigenvalue of the system matrix $\mxa$ corresponds to a characteristic system time of $T_{sys}\approx15s$.

\subsection{Verification of Theoretical Results}\label{ssec:pen2}

We first simulate the attack and defense process as described in Section~\ref{sec:met}, see \fig{exp1_angles}.
Using \eqref{eq:ratio4} we compute an attacker-optimal displacement vector $\vx_1$ for $t_1-t_0=t_2-t_1=T_{sys}$.
The attacker then performs a minimal energy attack to reach that state in time $[0,T_{sys}]$ during which the defender performs no action.
During time span $(T_{sys},2T_{sys}]$ the attacker does not perform any action but the defender uses minimum energy control to restore the steady state in the given time frame.
\fig{exp1_angles} shows the setting where both attacker and defender have access to all three pendula.

\begin{figure}
	\center
	\includegraphics[width=0.49\textwidth]{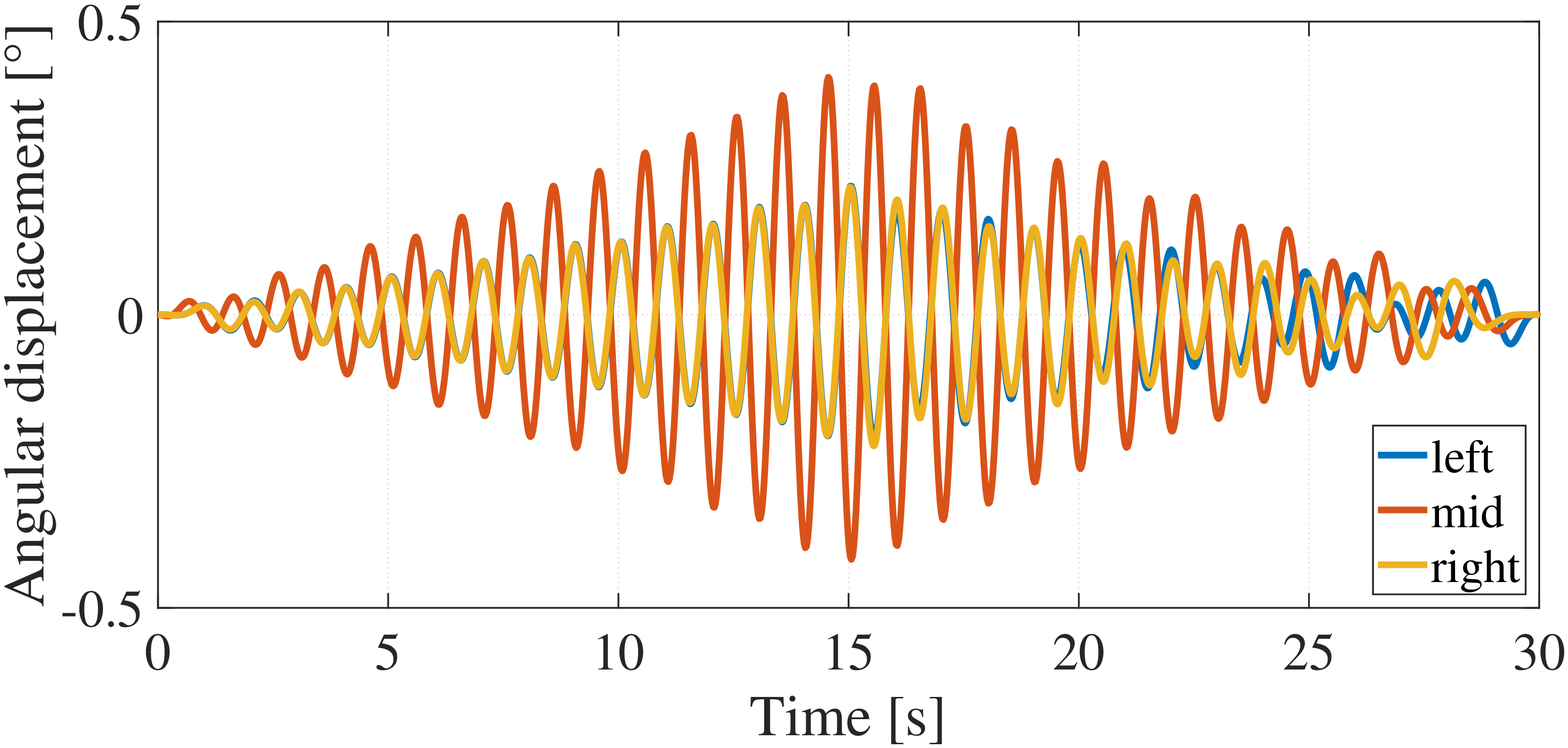}\\
	\includegraphics[width=0.49\textwidth]{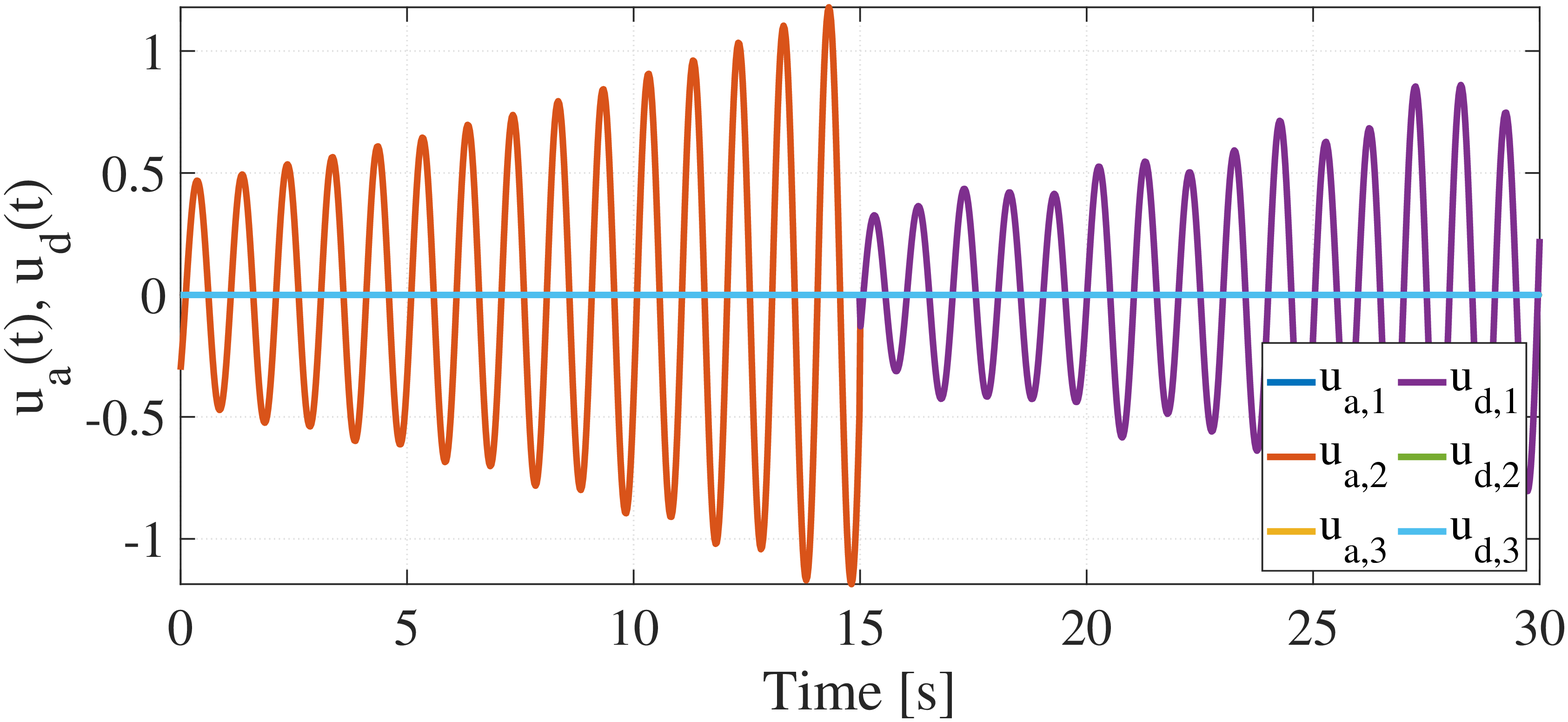}
	\caption{Angular displacements (top) and inputs (bottom) for the three pendula for an optimally chosen $15s$ minimum energy attack and a consecutive $15s$ minimum energy restoration phase.
		Attacker and defender have access to all three pendula here.}\label{fig:exp1_angles}
\end{figure}

\fig{exp1_angles} demonstrates that the maximum amplitudes are attained at $t=T_{sys}$ and that the (open loop) minimum energy control for the restoration phase actually brings the system state back to $\vx_0=0$ within the planned time frame.
The control energies were numerically computed for the input trajectories shown in \fig{exp1_angles}.
Their ratio matched the theoretical results predicted from \eqref{eq:ratio4} within numerical precision.
It is observable that the defensive control has only about half of the amplitudes for the inputs compared to the attacker. 
The resulting control energy for the defender is thus smaller than that of the attacker which is reflected in a computed optimal energy ratio $\rho\approx 7.32$.

\begin{table}
	\caption{Calculated resilience indices for the different attack and defender configurations}
	\label{tab:ex1a}
	\begin{center}
		\begin{tabularx}{0.49\textwidth}{|c|*4{>{\centering\arraybackslash}X|}}\hline
			
			\diagbox{Attacker}{Defender}& Left&Middle& Right&All\\\hline
	Left &6.79   & 0.04   & 6.85&   31.32\\\hline
Middle &1.80    &6.79   & 1.79 &  10.95\\\hline
Right &6.90    &0.04  &  6.79  & 31.63\\\hline
All &1.19    &0.02 &   1.19   & 7.32
\\\hline
		\end{tabularx}
	\end{center}
\end{table}

Table~\ref{tab:ex1a} shows the calculated resilience indices for all possible attacker/defender configurations.
Left / middle / right / all means that the attacker / defender has only access to the forces / torques for the left / middle / right / all pendulum.
This can be realized via choosing appropriate input matrices $\mxb_a$ and $\mxb_d$ consisting of zeros and ones only.

The values of the resilience indices for the ``left/left'',  ``right/right''  ``middle/middle'' configurations are close to maximal within the set of configurations with one attack and defense input only.  
Matching the attack vector is known to be an effective defender position \cite{Kang2009}.

Interestingly, the resilience indices of the configurations ``left/right'' and ``right/left'' are also numerically close to the matched ones. 
This can be explained by the approximate symmetry in the system, i.e., the left and right pendulum move very similarly if the middle pendulum is actuated. 
This insight does not directly follow from the attacker / defender placement but can be discovered with the new resilience index.
As intuitively to be expected, the resilience values are generally larger when the defender has access to all locations.

Overall, the resilience indices of the different situations allow to select plausible defense positions for  given attacks,
or the most damaging attack given a fixed defense structure. 
The latter result may be used to harden the access to the dynamic states.

In \fig{rho_log} we examine the dependence of the resilience index on the time spans $\Delta t$ for the attack / defense situation.
To this end, we show the results for $\rho(0,\Delta t, 2\Delta t)$ for all times $\Delta t$ from one order of magnitude smaller than $T_{sys}$ to one larger. 

\begin{figure}
	\includegraphics[width=0.49\textwidth]{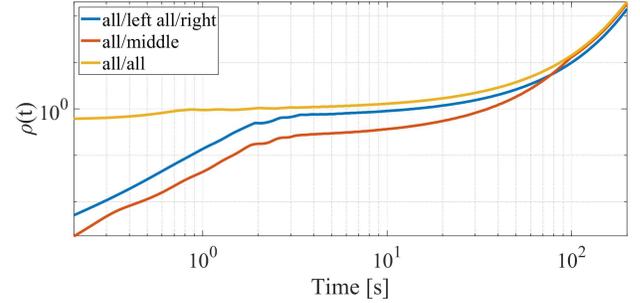}
	\caption{Resilience indices $\rho(0,\Delta t, 2\Delta t)$ for different times $\Delta t$ and different defender configurations. The attacker is able to actuate all pendula here. }\label{fig:rho_log}
\end{figure}

If $\Delta t$ is around the order of magnitude of the characteristic system time $T_{sys} = 15s$, plausible finite values can be observed for the resilience index.
The relative ordering for the different defense configurations is preserved in that range.

As expected indices converge to infinity for large $\Delta t$. 
The relative differences are becoming smaller for large $\Delta t$. This is because the defensive energy for large $\Delta t$ converges to zero due to the stable system dynamics. 
Whatever distorted state one starts from, the pendula will always return to rest without the need for much defender action.

For very small times $\Delta t$ the resilience indices are converging to zeros, if not all pendula are actuated by the defender.
This is because the defender needs to exploit the system dynamics to correct a distortion on a pendulum where it has no direct access.
Without any time for the system to evolve, this is not possible and the defender energy diverges to infinity.

\subsection{Extrapolation Tests}

We now test the predictive power of the resilience index
when the assumptions of Section~\ref{sec:met} are not exactly fulfilled.

Suppose that the defender does not follow an active monitoring and dedicated defense strategy but just applies a typical stabilizing feedback controller to counteract all types of disturbances whether malicious or not.
This is a typical situation in many real systems.
For the attacker, on the contrary, we assume that he or she is well-informed about the system and aims to achieve maximal damage at lowest possible effort as measured by the involved control energies.
Thus the minimum energy attack is still plausible here.

To demonstrate this scenario we implemented a continuously operating LQ feedback controller for the left pendulum as the defensive strategy changing the characteristic system time to $T_{sys}\approx 4.73~s$.
The attacker is allowed to access all nodes for its minimum energy attack.
This time, however, the attack is designed with information about the system's closed loop dynamics.
The attack duration to reach the state $\vx_1$ is $15~s$ and the system is monitored until $t=30~s$.

\fig{exp2_angles} shows the time evolution of the angles and inputs for the case ``all/left''.
Since the attack is designed with knowledge of the defender's control law and in context of the closed loop system, the attacker achieves its goal state $\vx_1$.
The maximum angular displacement can thus be observed at $t=15~s$.
The defender can be observed to be active during the whole time period.
The system is transferred back close to the steady state at time $t=30s$.

In Tab.~\ref{tab:ex2} the measured energy ratios for the different configurations are shown as well as the computed resilience indices.

Several observations can be made.
First, the values of the resilience index are much larger than the measured energy ratios.
This is because the resilience index assumes an optimal defense, but the LQ controller is not energy minimal and thus leads to a lower energy ratio.
Second, the relative size of the resilience index and the measured energy ratio for the different attack scenarios are similar, though not identical.
The attack on all nodes leads to the lowest resilience index, followed by an attack on the middle pendulum.
This demonstrates the transferability of the implications that would be taken from the novel resilience index -- e.g. protecting the middle pendulum is more important for a defender than defending the outer pendula.
\begin{figure}
	\center
	\includegraphics[width=0.49\textwidth]{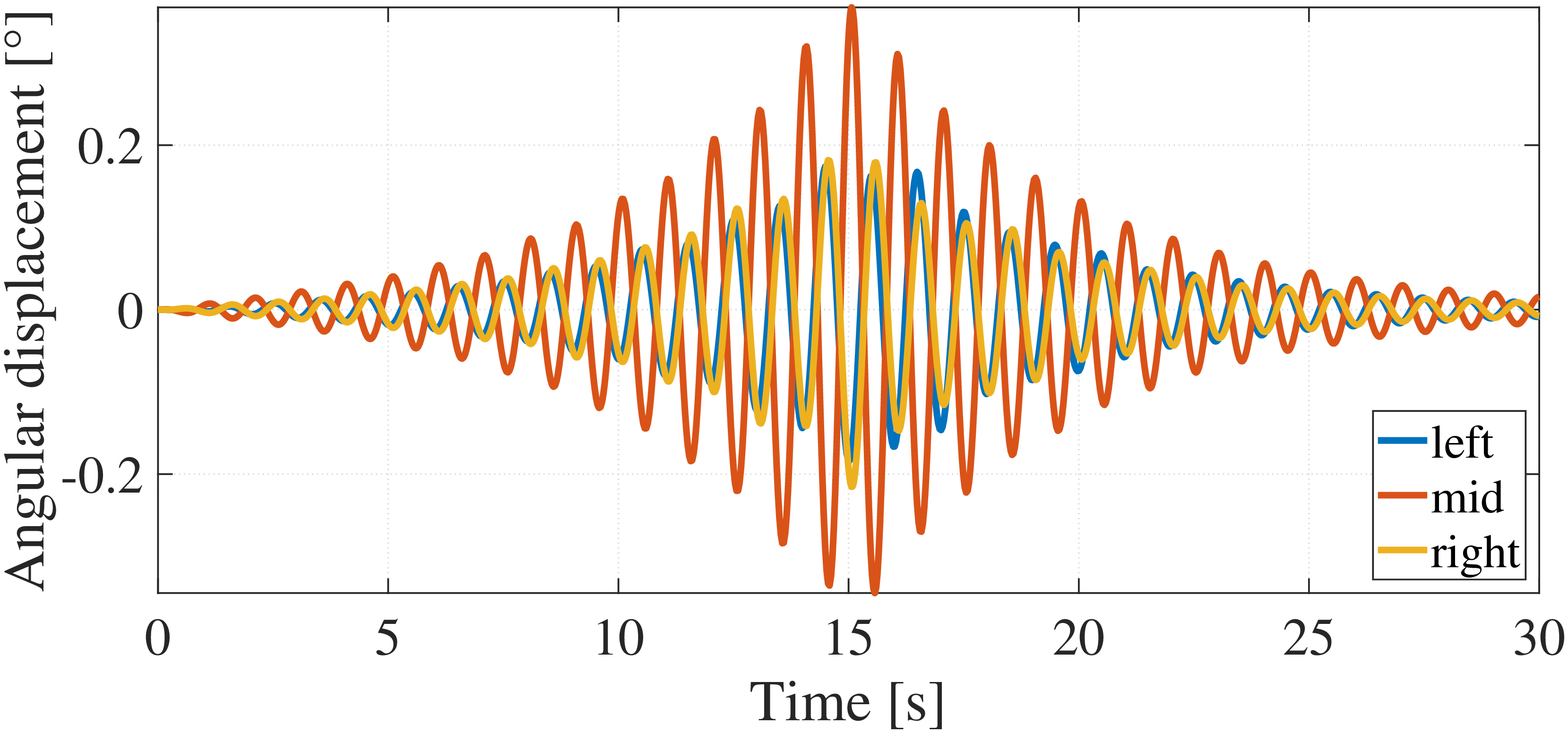}
	\includegraphics[width=0.49\textwidth]{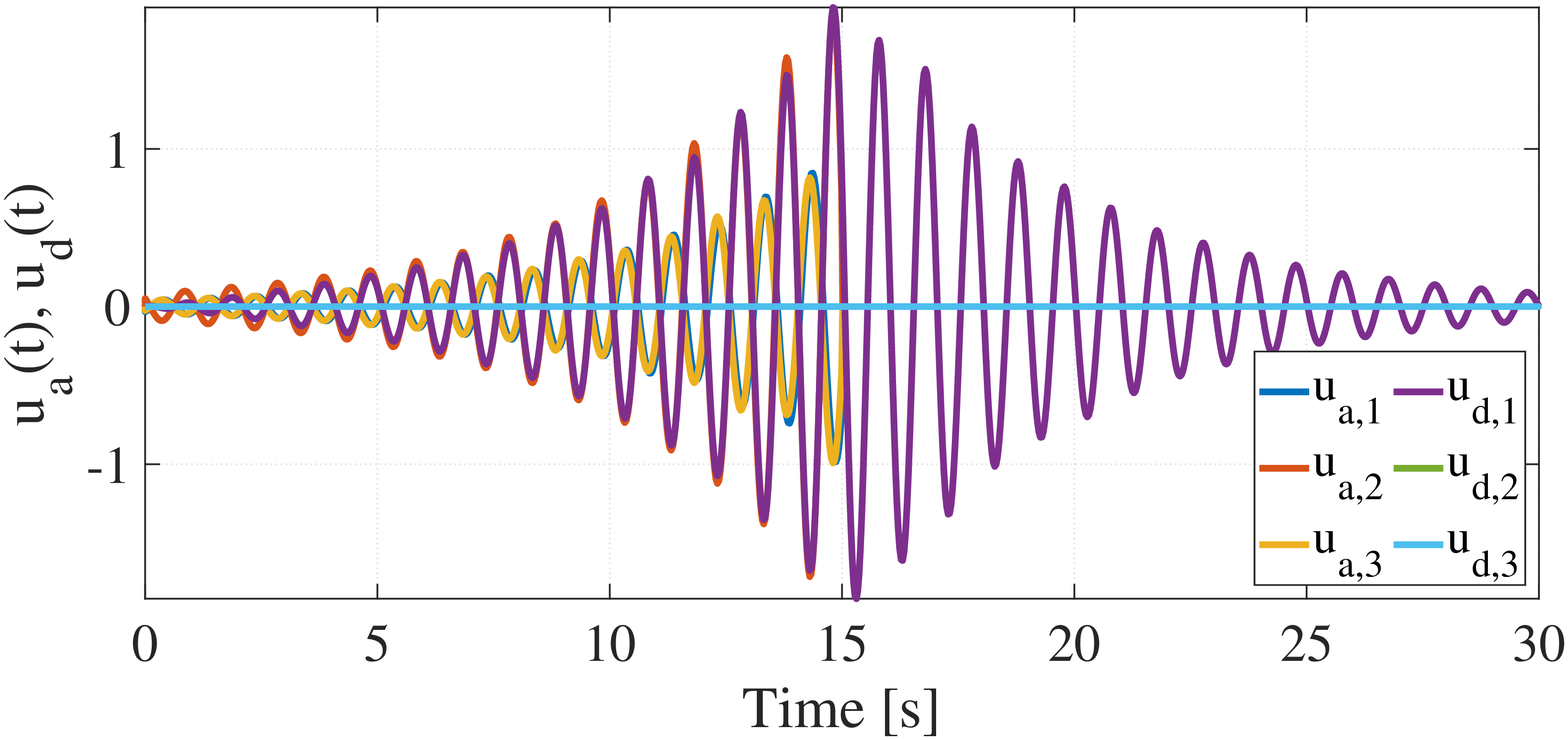}
	\caption{
		Angular displacements (top) and inputs (bottom) for the three pendula for an optimally chosen $15s$ minimum energy attack while a defensive LQ feedback controller is active throughout. 
		The attacker is assumed to have access to all pendula whereas the stabilizing LQ controller only actuates the left pendulum.
		}
	\label{fig:exp2_angles}
\end{figure}
\begin{table}
	\caption{Measured energy ratios and resilience indices for the extrapolation experiment}
	\label{tab:ex2}
	\begin{center}
		\begin{tabularx}{0.49\textwidth}{|c|*4{>{\centering\arraybackslash}X|}}\hline
			Attacker Configuration& Left&Middle& Right&All\\\hline
Measured Energy Ratio&3.53&1.08&3.95&0.73\\\hline
Resilience Index&465.68&114.89&464.67&77.07\\\hline
		\end{tabularx}
	\end{center}
\end{table}

\section{Conclusion}\label{sec:con}

A novel metric for LTI systems with attacked inputs and inputs available for a defensive controller is proposed.
It is based on the minimal ratio of the disturbing and restorative control energies.
It can be shown to be efficiently computable via a generalized eigenvalue problem involving the controllability Gramians of the attack and defensive input matrices.
The proposed approach extends the previous concept of the degree of disturbance rejection 
which also considers the control energies in the interplay between a specific type of distortion and the counteractions.
Our approach, however, does not focus on average efforts but on the worst-case, adversarial behavior of an attacker.
This is relevant e.g. when considering cyber attacks.

The metric enables system designers to compare different defensive control strategies for a given attack scenario or to analyze the largest vulnerability to plausible attackers, which again may be used to harden certain parts of the system.
Our experiments show that the resilience index not only covers the exact setup used to derive the index value, i.e., a consecutive minimal energy attack and defense, but also other plausible settings.
While the time span used to compute the index has to be considered carefully, 
the index can be used to uncover important vulnerabilities from the system which might not be obvious from the physical setup, as is demonstrated in our first experiment.

Further research will examine more real-world use cases and large scale settings that are often found in power system applications.

\appendix

\begin{lemm}
For symmetric positive definite matrices $\mxa$, $\mxb$ it is 
$$\min_{\vx} \frac{\vx^T\mxa^{-1}\vx}{\vx^T\mxb^{-1}\vx}
= \min_{\vx} \frac{\vx^T\mxb\vx}{\vx^T\mxa\vx} 
= \left(\max_{\vx} \frac{\vx^T\mxa\vx}{\vx^T\mxb\vx}\right)^{-1}.$$
Moreover if $\vx_l^*$ denotes the minimum-attaining vector of the left hand side and $\vx_{mr}^*$ of the middle and right hand side, respectively, then $\vx_l^* = \mxb \vx_{mr}^*$. 
\end{lemm}
\begin{proof}
It is
\eqna{\label{eq:R1}
\min_{\vx_1} \frac{\vx_1^T\mxa^{-1}\vx_1}{\vx_1^T\mxb^{-1}\vx_1}
= \min_{\vx} \frac{\vx^T\mxb^{\frac{1}{2}}\mxa^{-1}\mxb^{\frac{1}{2}}\vx}{\vx^T\vx}
}
which can be obtained by setting $\vx = \mxb^{-\frac{1}{2}}\vx_1$.
Similarly, we obtain
\eqna{\label{eq:R2}
\max_{\vx_2} \frac{\vx_2^T\mxa\vx_2}{\vx_2^T\mxb\vx_2}
= \min_{\vx} \frac{\vx^T\mxb^{-\frac{1}{2}}\mxa\mxb^{-\frac{1}{2}}\vx}{\vx^T\vx}
}
by setting $\vx = \mxb^{\frac{1}{2}}\vx_2$.
Since
$$\left(\mxb^{\frac{1}{2}}\mxa^{-1}\mxb^{\frac{1}{2}}\right)^{-1}
= \mxb^{-\frac{1}{2}}\mxa\mxb^{-\frac{1}{2}}$$
the two right hand side problems in \eqref{eq:R1} and \eqref{eq:R2} can both be solved by an eigenvalue decomposition of the positive definite matrix
$\mxb^{\frac{1}{2}}\mxa^{-1}\mxb^{\frac{1}{2}}$. 
Its minimal eigenvalue is the inverse of the maximal eigenvalue of  $\mxb^{-\frac{1}{2}}\mxa\mxb^{-\frac{1}{2}}$.
The corresponding two eigenvectors $\vx^*$ are identical.
For the original vectors $\vx_1^*$ and $\vx_2^*$ it follows that $\vx_1^* = \mxb^{\frac{1}{2}} \vx^* = \mxb \vx_2^*$.

\end{proof}

\addtolength{\textheight}{-12cm}   
\bibliographystyle{ieeetran}
\bibliography{Literature.bib}

\end{document}